%% file: Dhillon_WCL2016-0934.tex
\documentclass[journal]{IEEEtran}
 \usepackage{cite}

\ifCLASSINFOpdf
   \usepackage[pdftex]{graphicx}
   \graphicspath{{img/pdf/}{img/jpeg/}}
   \DeclareGraphicsExtensions{.pdf,.jpeg,.png}
\else
   \usepackage[dvips]{graphicx}
   \graphicspath{{img/eps/}}
   \DeclareGraphicsExtensions{.eps}
\fi

\usepackage[cmex10]{amsmath}
\usepackage{booktabs}
\usepackage{amsfonts}
\usepackage{amssymb}
\usepackage{mathrsfs}
\usepackage{bigints}
\usepackage{mathtools}

\usepackage[tight,footnotesize]{subfigure}
\usepackage{wasysym}
\usepackage{color}
\usepackage{epsfig} 
\usepackage{epstopdf}
\usepackage{float}
\usepackage{amsthm}
\usepackage{mathtools}
\usepackage{relsize}
\definecolor{light-gray}{gray}{0.8}
\let\emptyset\varnothing
\include{notation}

\newcommand{\Lagr}{\mathcal{L}}
\newcommand{\ud}{\, \mathrm{d}}
\newtheorem{defn}[thm]{Definition} 
\newtheorem{corollary}[cor]{Corollary}

\begin{document}
\title{\huge Effect of User Mobility on the Performance of Device-to-Device Networks with Distributed Caching}
\author{\IEEEauthorblockN{Shankar Krishnan and Harpreet S. Dhillon}
\thanks{The authors are with Wireless@VT, Department of ECE, Virginia Tech, Blacksburg, VA, USA. Email: \{kshank93, hdhillon\}@vt.edu.} \vspace{-2ex}}
\maketitle
\begin{abstract}
We consider a distributed caching device-to-device (D2D) network in which a user's file of interest is cached as several portions in the storage of other devices in the network. Assuming that the user needs to obtain all these file portions, the portions cached farther away naturally become the performance bottleneck. This is due to the fact that dominant interferers may be closer to the receiver than the serving device. Using a simple stochastic geometry model, we concretely demonstrate that this bottleneck can be loosened if the users are mobile. Gains obtained from mobility are quantified in terms of coverage probability.
\end{abstract}
\IEEEpeerreviewmaketitle
\allowdisplaybreaks
\section{Introduction}

Caching popular content on the user devices and delivering it asynchronously to other proximate devices via D2D communication help offload traffic from cellular network \cite{Magazine,Femto2}. However, with the increasing popularity of ultra-HD videos (larger file sizes), it will become increasingly difficult to cache the entire file of interest in a single mobile device. This has led to the consideration of {\em distributed storage regime} \cite{altman2014distributed}, where the file of interest is stored as multiple portions among different devices in the network. In this letter, we focus on the performance analysis of such a system from wireless communications perspective. 
In particular, we focus on the fact that the file portions cached geographically farther from the receiver of interest  may not be easy to receive due to the presence of stronger interferers located closer than the serving device \cite{DistCaching}. Using a simple stochastic geometry model, we show that user mobility helps in dealing with this bottleneck. Exact gains are quantified in terms of coverage probability.

Prior works focusing on user mobility in cache-enabled D2D networks such as \cite{R2} usually consider simplistic mobility models or ignore interference from other transmitting D2D nodes (not always realistic) to make the analysis tractable. Using tools from stochastic geometry, there are other recent works that deal with the mobility-aware analysis of massive random networks (not particularly D2D networks). For example, \cite{Gong} showed that user mobility helps reduce the local delay in Poisson networks and \cite{LinMulticast} showed that mobility increases the mean number of {\em covered} receivers in a multicast D2D network. Building on these works, we develop a simple stochastic geometry model to enable mobility-aware analysis of distributed caching in D2D networks. The key difference behind our and these existing works is that we provide new analytical results on key metrics such as coverage probability by capturing the user's local neighborhood as it moves around in a distributed caching network. To expose fundamental design insights, we consider a $2$-file portion distributed caching system, where a typical user successfully receives one file portion at its initial location and receives the other file portion at its next location, which is assumed to be distance $v\geq 0$ away. Modeling the device locations as a Poisson Point Process (PPP) \cite{HaeB2013}, we first derive distance distributions for receiving the second file portion (termed {\em farther} file portion) at the second location. The exact analysis at the second location is not straightforward and requires the knowledge of the {\em local neighborhood} as observed at the first location. After carefully incorporating this information in the form of asymmetric {\em exclusion zone} with respect to the second location, we derive tractable expressions for the coverage probability of receiving the farther cached file portion (the one that was not received at the first location) for different levels of mobility. Our results concretely demonstrate that coverage probability at the second location increases with user mobility and asymptotically approaches an \emph {independent} scenario, where the coverage probability of a file portion is independent of its geographical location in the network.

\section{System Model}
\subsubsection*{System Setup}
Device locations are modeled as a homogeneous PPP $\Phi$ with intensity $\lambda$. We consider a $2$-file portion distributed caching system where each device has either file A or file B cached independently with probabilities $p_A = p$ and $p_B = 1-p$. Here file A and B correspond to two portions of a larger file requested by the {\em typical} device. Independent thinning of the original PPP $\Phi$ results in two independent PPPs, $\Phi_A$ for file A and $\Phi_B$ for file B. Conditioned on the serving link, each interferer is assumed to be active independently with probability $q$. This factor captures the fact that all the devices in the network may not always be active.



\subsubsection*{Channel Model}
For the wireless channels, we assume distance-dependent power-law pathloss with exponent $\alpha$ and Rayleigh fading. For a typical user located at the origin, the power received at the user from a device $x \in \Phi$ is $P = P_t h_x \|x\|^{-\alpha}$, where $P_t$ is the transmit power, $h_x \sim \exp(1)$ models Rayleigh fading, and $\alpha>2$ is the pathloss exponent. To define the interference power from any node $y$, we need an additional binary random variable $t_y$, which takes value $1$ with probability $q$ and $0$ otherwise. For this setup, the received signal to interference ratio ($\sir$) at the typical device can be expressed as
\begin{align}
\sir =    \frac{P_t h_x \|x\|^{-\alpha}}{\sum\limits_{y \in \Phi \backslash \{x\}} t_y P_t h_y \|y\|^{-\alpha} } = \frac{h_x \|x\|^{-\alpha}}{\sum\limits_{y \in \Phi \backslash \{x\}} t_y h_y \|y\|^{-\alpha} }.
\end{align}
We consider out-of-band D2D due to which the interference from the cellular network does not show up in the received $\sir$. Quite reasonably, the network is assumed to be interference limited, as a result of which the thermal noise is ignored in comparison to the interference power.
\begin{figure}[t!]
\centering{
\includegraphics[width=.9\linewidth]{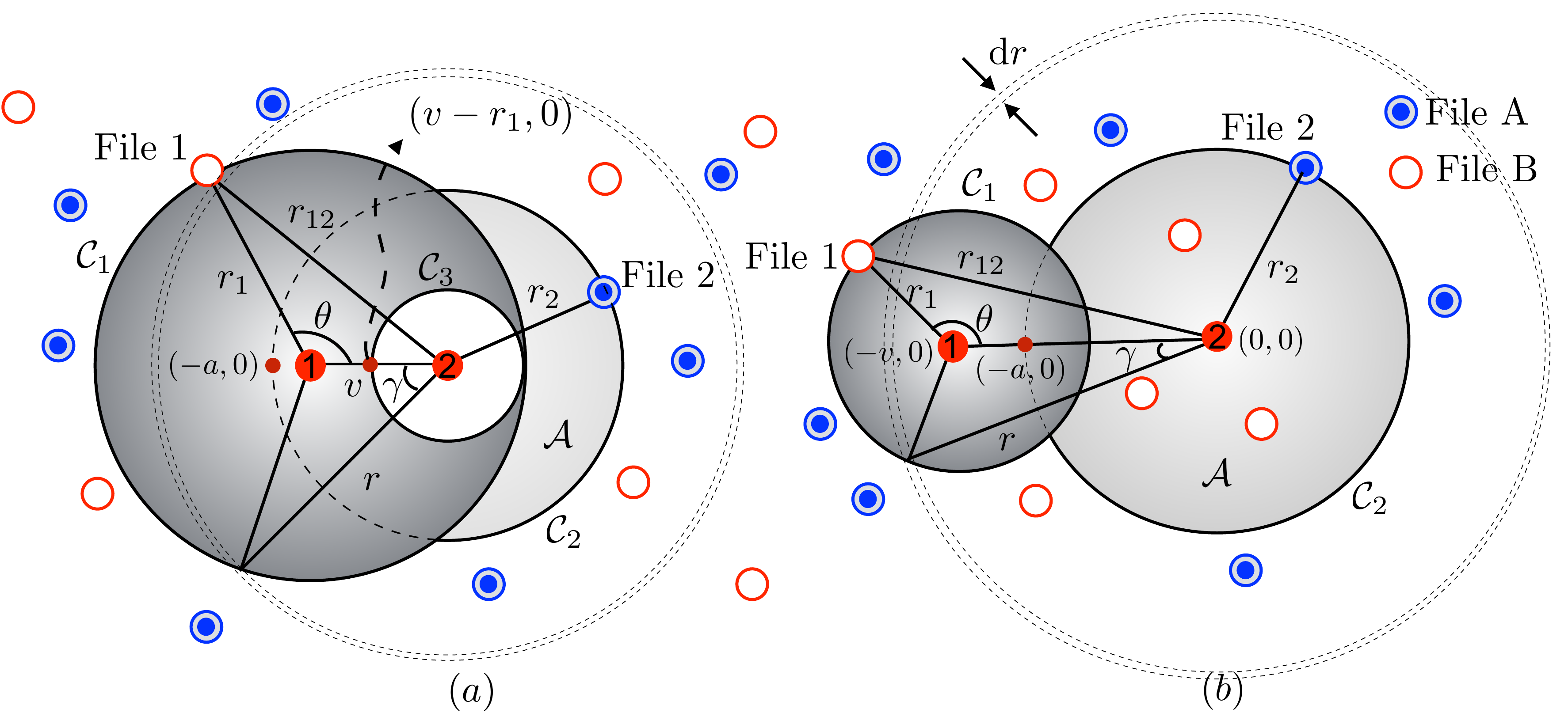}
\caption{System model. (\emph{a}) Scenario 1 ($v <r_1$) and (\emph{b}) Scenario 2 ($v > r_1$). A user at location 1 $(-v,0)$ receives file 1 and moves a distance $v$ to location 2 $(0,0)$, where it receives file 2. Subcase $\mathcal{Y}$ is shown (File B is File 1).}
\label{Fig:Circle Model}
}
\end{figure}
\section {Effect of User Mobility on Coverage}
A user initially located at $(-v,0)$, termed $\emph{location 1}$, connects to its closest device and successfully receives the file portion cached by that device (can be file A or B). 
This user then moves distance $v$ to \emph{location 2} (taken to be the origin) and receives the other file portion from the closest device that has the other file portion in its cache. Our goal is to study the coverage probability at \emph{location 2} as a function of $v$ ($v=0$ models {\em static} case studied in \cite{DistCaching}). 
We label the closest file portion cached to the user at location 1 (can be file A or B) as \emph{file 1} and the other file portion as \emph{file 2}. Fig. \ref{Fig:Circle Model} depicts this system setup with $R_1$ and $R_2$ denoting the distances of the user at locations 1 and 2 to the closest device with file 1 and file 2, respectively ($r_1$, $r_2$ denote realizations of $R_1$, $R_2$). Also let us define two circles: $\mathcal{C}_1$ centered at location 1 with radius $r_1$ and $\mathcal{C}_2$ centered at location 2 with radius $r_2$. When user moves a distance $v$ from location 1 to 2, one of the following three cases arises: (i) {\em Case 1:} Disjoint circles ($v \geqslant r_1+r_2$), (ii) {\em Case 2:} Intersecting circles ($ r_2 - r_1 < v < r_1 +r_2 $), and (iii) {\em Case 3:} Engulfed circles ($v \leqslant r_2 - r_1$). Following intermediate result will be useful for our analysis.
\begin{defn} Consider two partially-overlapping circles with radii $r_1$ and $r_2$ with centers separated by distance $v$, where $ r_2 - r_1 < v < r_1 +r_2$, as shown in Fig. \ref{Fig:Circle Model}. The lightly shaded region $\mathcal{A}$ is called a lune and its area is~\cite[Equation (12.76)]{lunebook}
\begin{align}
&\notag \mathcal{A}_{{\tt lune}} = \pi r_2^2  +\frac{1}{2}\sqrt{\big[(r_1 +v)^2 - r_2^2\big]\big[r_2^2-(r_1-v)^2\big]} \\\notag & - r_1^2 \cos^{-1} \left(\frac{r_1^2+v^2-r_2^2}{ 2vr_1}\right) - r_2^2 \cos^{-1} \left(\frac{r_2^2+v^2-r_1^2}{2vr_2}\right). 
\end{align}
\end{defn}
\subsection{Distance distribution}
Recall that $R_1$ is the distance from location 1 to the closest point of a PPP with intensity $\lambda$. Its distribution can be found from the null probability of a PPP as 
$f_{R_1}(r_1) = 2\pi \lambda r_1 e^{- \lambda \pi r_1^2}$~\cite{HaeB2013}. For the system setup studied in this paper, there exist two possible subcases: (i) $\mathcal{X}$: File A is file 1 (occurs with probability $p_A$) and (ii) $\mathcal{Y}$: File B is file 1 (occurs with probability $p_B$). The subcase $\mathcal{Y}$ is depicted in Fig. \ref{Fig:Circle Model}.

As there exists no device cached with either file portion within a distance $r_1$ from location 1, $\mathcal{C}_1$ can be interpreted as an {\em exclusion zone} in the interference field. Conditioned on a certain file portion located at a distance $r_1$, the distribution of $R_2$ is hence dictated by the presence of no devices caching the other file portion in $\mathcal{C}_2 \setminus \mathcal{C}_1$. In Fig. \ref{Fig:Circle Model}, the lightly shaded region $\mathcal{A}$ represents $\mathcal{C}_2 \setminus \mathcal{C}_1$, which depends on the distance $v$ moved by the user between location 1 and 2. As per the definition of the three cases, $\mathcal{A}$ is represented by the entire circle $\mathcal{C}_2$ in case of  disjoint circles (case 1), a \emph{lune} in case of intersecting circles (case 2) or an annular region between circles $\mathcal{C}_1$ and $\mathcal{C}_2$ in case of engulfed circles (case 3). The area of $\mathcal{A}$ is mathematically expressed in Lemma \ref{Lem:Distribution}.

Before stating Lemma \ref{Lem:Distribution}, it is worth defining two scenarios: termed \emph{scenarios 1} and {\em  2}, based on the distance $v$ moved by the user from location 1 to 2. In scenario 1, distance $v$ moved by the user is smaller than the serving distance at location 1 i.e. $v < r_1$. As a result, the user is still inside circle $\mathcal{C}_1$ (Fig.\ref{Fig:Circle Model} \emph{(a)}) and hence no device can lie within a distance $r_1 - v$ from location 2. Thus, the closest device with file 2 is located atleast a distance of $r_1 - v$ away from location 2 ($r_2 \geq r_1-v$). In scenario 2, the user moves a larger distance ($v > r_1$), as a result of which the user moves out of the circle $\mathcal{C}_1$ (Fig.\ref{Fig:Circle Model} \emph{(b)}). Hence there exists no such condition for the serving distance $R_2$ of file 2 in scenario 2. The above mathematical conditions for $R_2$ based on the two scenarios are handled appropriately by defining $z_1 = {\tt max} (0,r_1 - v)$ as its lower limit. Also let us define a circle $\mathcal{C}_3$ centered at location 2 and a radius of $z_1$, which will be used later in the analysis. It is to be noted that $\mathcal{C}_3$ converges to a point for scenario 2 ($v > r_1$) and hence not shown in the corresponding figure. The conditional distribution of $R_2$ is now derived next in Lemma \ref{Lem:Distribution}.
\begin{lemma} \label{Lem:Distribution}
For a given $v$, the conditional distribution of distance $R_2$ from location 2, conditioned on $r_1$ is
\begin{align}
\notag f_{R_2 | R_1}(r_2|r_1) &= \frac{\mathrm{d}}{\mathrm{d}r_2}({1-e^{-\lambda_2 | \mathcal{A}|}}) , \qquad  r_2 \geq {\tt max}(0,r_1-v) \\\notag
|\mathcal{A}| &= \left\{
     \begin{array}{lr}
       \pi r_2^2,  \qquad {\tt max}(0,r_1-v) \leq r_2 \leq |v-r_1|\\
       \mathcal{A}_{\tt lune},   \qquad \qquad |v-r_1| < r_2 < v+r_1\\
       \pi(r_2^2 -r_1^2), \qquad   r_2 \geq v+r_1
      \end{array}
   \right.
\end{align}
where $\lambda_2 = p_A \lambda$ with probability $p_A$ and $p_B\lambda$ otherwise.
\end{lemma}
\begin{proof}
Conditioned on the presence of no device caching the other file portion (PPP of intensity $\lambda_2$) within a distance $r_1$ from location 1 (or equivalently in $\mathcal{C}_1$), the complementary cumulative distribution function (CCDF) of $R_2$ is given by
\begin{align*}
\bar{F}_{R_2|R_1}(r_2|r_1) &= \P\big(N(|\mathcal{C}_2|)=0\big|N(|\mathcal{C}_1|)=0\big) \\\notag & = \P({N}(|\mathcal{C}_2\setminus  \mathcal{C}_1|) = 0) \stackrel{(a)}= \exp(-\lambda_2 |\mathcal{A}|)
\end{align*}
where $|.|$ denotes the area, $N(.)$ is the number of files of other type in the specified area, and (a) results from the null probability of a PPP with intensity $\lambda_2$. The result follows by differentiating CCDF and using appropriate values for $|\mathcal{A}|$. 
\end{proof}
\subsection{Coverage probability of file 2} A user is said to be in coverage of a certain file portion if the received $\sir$ at that user from a device caching that file portion is greater than a given threshold $T$ i.e. coverage probability $P_c = \P(\sir > T)$. For the coverage probability analysis of file 2, we just focus on subcase $\mathcal{Y}$ (file B is file 1). Result for subcase $\mathcal{X}$ will follow immediately by swapping the variables. The total coverage probability of obtaining file 2 is derived by applying total probability theorem to the two subcases as $P_{c_2} = p_AP_{c_2}^{(\mathcal{X})} + p_BP_{c_2}^{(\mathcal{Y})}$, where $P_{c_2}^{(\mathcal{X})}$ and $P_{c_2}^{(\mathcal{Y})}$ denote the conditional coverage probability of obtaining file 2 in subcases $\mathcal{X}$ and $\mathcal{Y}$, respectively. Conditioned on $R_1$ and $R_2$, the coverage probability $P_{c_2}^{(\mathcal{Y})}$ can be determined by dividing the total interference field into three regions as described next.
\\$I_1$: Interference experienced at location 2 due to the transmission of the device $x$ $\in$ $\Phi_B$ (has file B) that was the {\em serving} device for location 1. As shown in Fig.~\ref{Fig:Circle Model}, this device is at distance $r_{12}$ from location 2. The interference power is
\begin{align}
I_1 &= t_xh_xr_{12}^{-\alpha}.
\label{eq:I1}
\end{align}
$I_2$: Interference at location 2 from all devices with file B except the singleton $\{x\}$ at distance $r_{12}$. This interference field is essentially $\Phi_B$ with an {\em asymmetric} exclusion zone $\mathcal{C}_1$ created by the exclusion of $\{x\}$. The interference power is
\begin{align}
I_2 &= \sum\limits_{y \in \Phi_B \backslash \mathcal{C}_1 } t_y h_y \|y\|^{-\alpha}.
\label{eq:I2}
\end{align}
$I_3$: Interference at location 2 from all devices with file A except the serving device from $\Phi_A$ at distance $r_2$. As there exists no device with file A in $\mathcal{C}_1$, this interference is equivalent to considering interference from $\Phi_A$ outside {\em exclusion zone} $\mathcal{C}_1\cup\mathcal{C}_2$.  The interference power is
\begin{align}
I_3 &= \sum\limits_{z \in \Phi_A \backslash (\mathcal{C}_1 \cup \mathcal{C}_2)} t_z h_z \|z\|^{-\alpha}.
\label{eq:I3}
\end{align}
Due to the Rayleigh fading assumption, coverage probability in general can be expressed in terms of the Laplace transform of the interference power distribution \cite{HaeB2013,andrews2011tractable}. Due to the independence of the three interference terms defined above, the Laplace transform of the distribution of total interference can be expressed as the product of the Laplace transforms of the three terms. Using this, the coverage probability of obtaining file 2 in subcase $\mathcal{Y}$ can be expressed as follows.
\begin{theorem} \label{thm:1}
The coverage probability of obtaining file 2 at location 2 in a PPP of intensity $\lambda$ for subcase $\mathcal{Y}$ is
\begin{align*}
&P_{c_2}^{(\mathcal{Y})} = \int_{0}^{\infty } \int_{z_1}^{\infty } \int_{0}^{\pi } \Lagr_{I_1|R_1,\Theta}(T{r_2}^\alpha| r_1,\theta) \Lagr_{I_2|R_1}(T{r_2}^\alpha | r_1) \\
&\Lagr_{I_3|R_1,R_2}(T{r_2}^\alpha | r_1,r_2) f_{R_2|R_1}(r_2|r_1)  f_{R_1}(r_1) f_{\Theta}(\theta) \ud \theta \ud r_2 \ud r_1
\end{align*}
where the conditional Laplace transforms of $I_1$, $I_2$ and $I_3$ are derived below in Lemmas \ref{Lem: Lap_I1}, \ref{Lem: Lap_I2} and \ref{Lem: Lap_I3} respectively.
\end{theorem}
\begin{proof} From the definition of coverage probability, 
\begin{align*}
& P_{c_2}^{(\mathcal{Y})} = \E_{R_1, R_2,\Theta}\big[\P(\sir > T| r_1, r_2,\theta) \big] \\ &= \E_{R_1, R_2,\Theta}\big[P\big(hr_2^{- \alpha} > T(I_1 + I_2 + I_3) \left| \right. r_1, r_2,\theta\big) \big] \\
& \stackrel{(a)}=  \E_{R_1, R_2,\Theta}\big[\E[e^{-s(I_1+I_2+I_3)}| r_1,r_2,\theta]\big] \\
& \stackrel{(b)}= \int_{0}^{\infty } \int_{z_1}^{\infty } \int_{0}^{\pi } \E[e^{-sI_1}|r_1,\theta] \E[e^{-sI_2}|r_1]\E[e^{-sI_3}|r_1,r_2] \\\notag & \qquad \qquad f_{R_2|R_1}(r_2|r_1)  f_{R_1}(r_1) f_{\Theta}(\theta) \ud \theta \ud r_2 \ud r_1
\end{align*}
where (a) results from $h \sim \exp(1)$ and defining $s = Tr_2^{\alpha}$. Step (b) follows from the independence of the three interference powers and deconditioning w.r.t. $R_1$, $R_2$ and $\Theta$, where $\Theta$ is a uniform random variable in $[0,\pi]$ i.e. $f_{\Theta}(\theta) = 1/\pi$. From (\ref{eq:I1}), (\ref{eq:I2}) and (\ref{eq:I3}), it can be seen that while $I_1$ and $I_2$ depend on just $r_1$, $I_3$ is a function of both $r_1$ and $r_2$. The result now follows by using the definition of Laplace transform for the interference powers and conditioning them accordingly.
\end{proof}
The Laplace transform of the distribution of interference $I_1$ from a {\em singleton} is dealt first in the following Lemma.
\begin{lemma} \label{Lem: Lap_I1}
Given $v$, the conditional Laplace transform of interference $I_1$ from a {\em singleton} defined in \eqref{eq:I1} is 
\begin{align*}
\Lagr_{I_1|R_1,\Theta}(s|r_1,\theta)= 1-q + \frac{q}{1+s(r_1^2+v^2-2r_1v\cos\theta)^{\frac{-\alpha}{2}} }.
\end{align*}
\end{lemma}
\begin{proof}
By definition, the Laplace transform of interference is
\begin{align*}
&\Lagr_{I_1|R_1,\Theta}(s|r_1,\theta) =  \E  [{e^{-st_xh_xr_{12}^{-\alpha}}}] \\\notag
& \stackrel{(a)} =  1-q + q\E [{e^{-sh_xr_{12}^{-\alpha}}}]     \stackrel{(b)} =  1-q +  \frac{q}{1+sr_{12}^{-\alpha }},     \notag
\end{align*}
where (a) follows from the fact that the interferer $ x \in \Phi_B$ located at $r_{12}$ is active with a probability $q$, and (b) results from $h_x \sim \exp(1)$. The final result follows by using the law of cosines in which $r_{12}^2 = r_1^2+v^2-2r_1v\cos\theta$ (see Fig.~\ref{Fig:Circle Model}).
\end{proof}
The conditional Laplace transform of interference $I_2$ is derived next with its proof provided in Appendix A. The key is in handling the asymmetric exclusion zone $\mathcal{C}_1$ carefully. 
Interested readers can refer to \cite{PHP} for more details on how such exclusion zones can be handled. 


\begin{lemma}\label{Lem: Lap_I2} Given $v$, the conditional Laplace transform of $I_2$ under subcase $\mathcal{Y}$ defined in \eqref{eq:I2} is $\Lagr_{I_{2}|R_1}(s|r_1)=$
\begin{align*}
&\exp\bigg(-2p_Bq\lambda \bigg(\: \int\limits_{z_1}^{\infty} \frac{\pi r \ud r }{1+\frac{r^{\alpha }}{s}}   - \int\limits_{|v-r_1|}^{v+r_1} \frac{f(r,r_1) r \ud r }{1+\frac{r^{\alpha }}{s}} \bigg)\bigg), \\
&\text{where} \: z_1 =  {\tt max}(0,r_1-v) \: , f(r,r_1) = \cos^{-1}\left(\frac{r^2+v^2-r_1^2}{2rv}\right).
\end{align*}
\end{lemma}
The conditional Laplace transform of interference $I_3$ is computed similarly with its proof  provided in Appendix B.
\begin{lemma}\label{Lem: Lap_I3} Given $v$, the conditional Laplace transform of $I_3$ under subcase $\mathcal{Y}$ is given by 
\begin{align}
\Lagr_{I_{3}|R_1,R_2}(s|r_1,r_2) &= \exp\bigg(- p_Aq\lambda \int\limits_{r_2}^{\infty} \frac{2\pi r \ud r}{1+\frac{r^{\alpha }}{s}} \bigg) \notag \\
& \exp\bigg(p_Aq\lambda \mathcal{B}(r_1,r_2,v)\bigg)\bigg), \notag
\end{align}
where $\mathcal{B}(r_1,r_2,v)$ is given by (\ref{Eqn:B_integral}).
\end{lemma}

Using the above results, we can also study the asymptotic coverage probability of file 2 when locations 1 and 2 are far apart ($v \to \infty$). In this case, $I_1 \rightarrow 0$, and the asymmetric exclusion zone $\mathcal{C}_1$ does not appear in $I_2$ and $I_3$. 
 
\begin{corollary}
For large user mobility ($v \rightarrow \infty$), the asymptotic coverage probability of file 2 is  given by:
\begin{align}
P_{c_2} &\to p_A{\tt P_c}(p_B) + p_B{\tt P_c}(p_A),
\end{align}
where ${\tt P_c}(p) = \frac{p}{p+q[\rho_1(T,\alpha)+(1-p)\rho_2(T,\alpha)]}$, 
$\rho_1(T,\alpha) = T^{\frac{2}{\alpha}} \int_{T^{-\frac{2}{\alpha}}}^{\infty} \frac{ \ud u}{1+u^{\frac{\alpha}{2}}}$, and $\rho_2(T,\alpha) = T^{\frac{2}{\alpha}} \int_0^{T^{-\frac{2}{\alpha}}} \frac{\ud u}{1+u^{\frac{\alpha}{2}}}$. 
\begin{proof}
For subcase $\mathcal{Y}$ and $v \rightarrow \infty$,  $\Lagr_{I_1|R_1,\Theta}(Tr_2^{\alpha}|r_1,\theta) = 1$,  $\Lagr_{I_{2}|R_1}(Tr_2^{\alpha}|r_1)=  \exp\big(-2\pi p_Bq\lambda  \int\limits_{0}^{\infty} \frac{ r \ud r }{1+\frac{r^{\alpha }}{Tr_2^{\alpha}}}\big)$ and  $\Lagr_{I_{3}|R_1,R_2}(Tr_2^{\alpha}|r_1,r_2)= \exp\big(-2\pi p_Aq\lambda  \int\limits_{r_2}^{\infty} \frac{ r \ud r }{1+\frac{r^{\alpha }}{Tr_2^{\alpha}}}\big)$. Also $v \rightarrow \infty$ corresponds to case 1 (disjoint circles) defined in our setup and thus $f_{R_2 | R_1}(r_2|r_1) = 2 p_A \lambda \pi r_2 \exp(-p_A \lambda \pi r_2^2)$ (from Lemma 1). Plugging the above values in Theorem \ref{thm:1}, we obtain $P_{c_2}^{(\mathcal{Y})} = {\tt P_c}(p_A)$. Similarly, it can be shown that $P_{c_2}^{(\mathcal{X})} = {\tt P_c}(p_B)$ for $v \rightarrow \infty$. The final result follows by applying total probability theorem to the two subcases.
\end{proof}
\end{corollary}

\section{Results and Discussion}

Fig. \ref{Fig:EffectMob} plots the coverage probability of file $2$ for various values of $v$. Consistent with intuition, coverage probability improved with increasing mobility. For a static user, the low coverage of file 2 is due to the presence of a dominant interferer (file 1) located closer to the user than the serving device (file 2). With user mobility, the likelihood of having a dominant interferer that is located closer to the serving device reduces, which improves the coverage of file 2. 


\begin{figure}[t!]
\centering
\begin{subfigure}{}
\centering
\includegraphics[width=0.47 \linewidth]{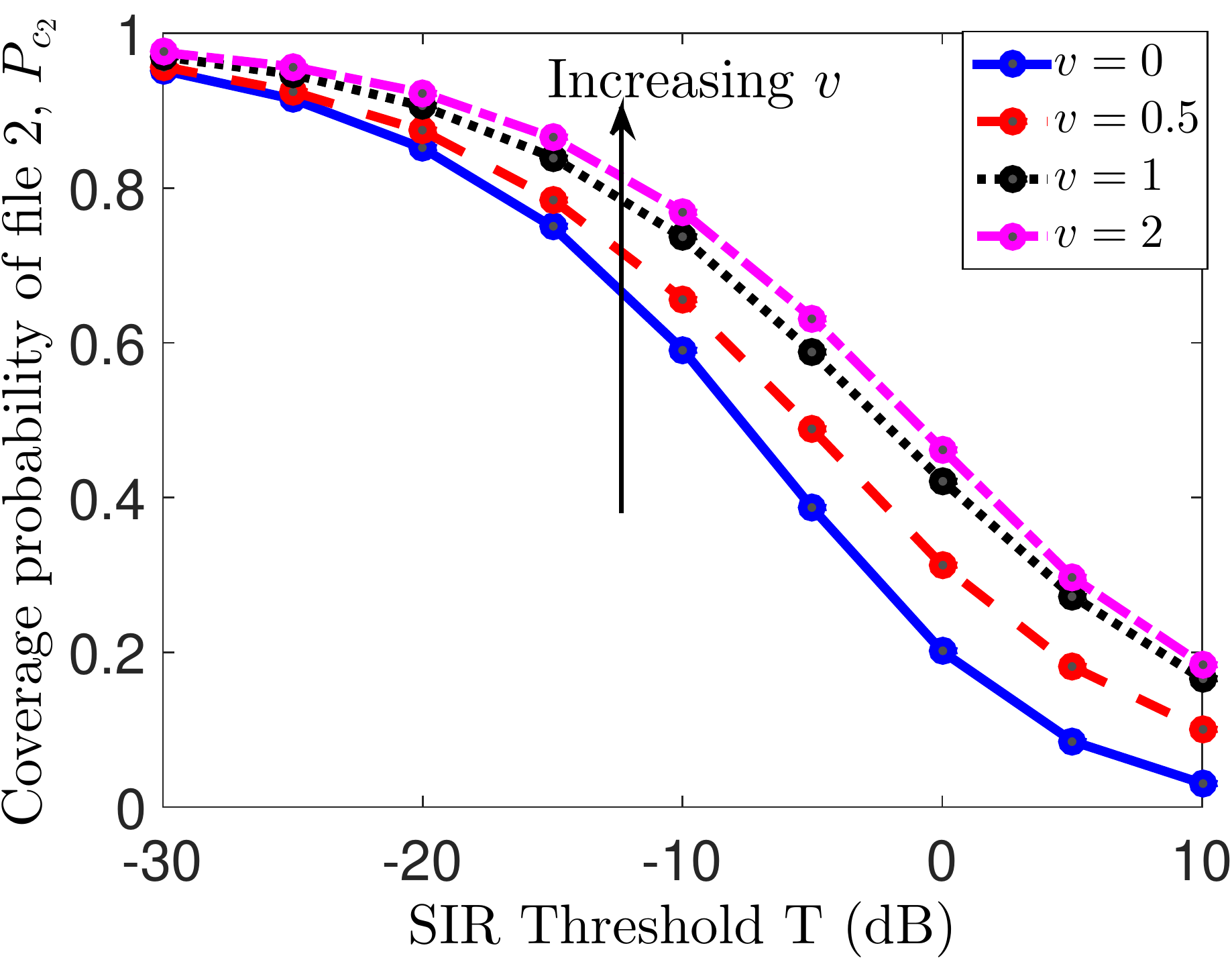}
\label{Fig:wrtT}
\end{subfigure}
\begin{subfigure}{}
\centering
\includegraphics[width=0.47 \linewidth]{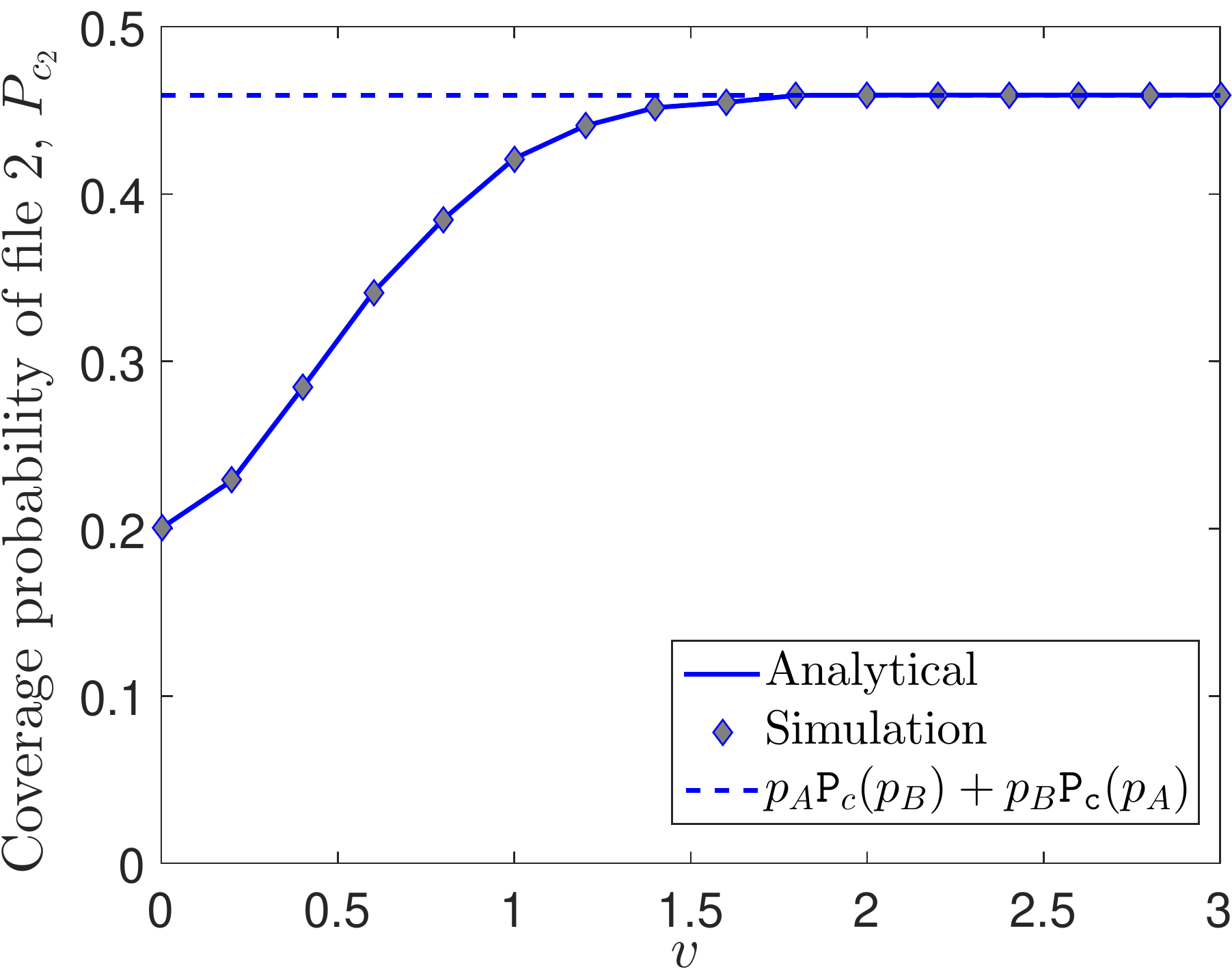}
\label{Fig:Asym}
\end{subfigure}
\caption{ Coverage probability of file 2 ($p_A = 0.5, q = 0.5, \alpha = 4$, $\lambda = 1$) for i) {\em (left)} varying $T$ and ii) {\em (right)} varying $v$ ($T = 0$ dB).}
\label{Fig:EffectMob}
\end{figure}

\section{Conclusion}
In this paper, we show that user mobility improves the coverage probability of obtaining the farther cached file portion in a 2-file portion distributed caching D2D network. This improvement is a result of lower likelihood of having dominant interferer closer than the serving device for a mobile user. The novelty of this work lies in the {\em exact} mobility-aware analysis that has to capture the information of the local neighborhood of nodes at the original location of the device. Due to the complexity of the analyses, the extension to multiple file portions (more than two) appears tedious but nonetheless forms an interesting avenue for further investigation. Also, one could extend this work to study the effect of mobility for more advanced mobility models and possibly an in-band D2D system as well.

\appendix
\subsection{Proof of Lemma \ref{Lem: Lap_I2}}
The conditional Laplace transform of interference $I_2$ (all file B devices except the $\emph{singleton}$)  is written as 
\begin{align}
&\notag \Lagr_{I_{2}|R_1}(s|r_1) = \E [e^{-sI_2}]
\stackrel{(a)}= \E \bigg[\prod\limits_{y \in \Phi_B \backslash \mathcal{C}_1} 1-q + \frac{q}{1+s{\|y\|}^{-\alpha }}\bigg]\\\notag
&\stackrel{(b)}= e^{-p_Bq\lambda \int\limits_{\R^2 \backslash (\mathcal{C}_3 \cup (\mathcal{C}_1 \setminus \mathcal{C}_3))  } \frac{\ud y }{1+\frac{{\|y\|}^{\alpha }}{s}}} \\\notag
&\stackrel{(c)}= e^{-p_Bq\lambda \left(\: \int\limits_{z_1}^{\infty} \frac{2\pi r \ud r }{1+\frac{r^{\alpha }}{s}}   -  \int\limits_{|v-r_1|}^{v+r_1} \frac{ 2\cos^{-1}\big(\frac{r^2+v^2-r_1^2}{2rv}\big) r \ud r }{1+\frac{r^{\alpha }}{s}}   \right)} \notag
\end{align}
where (a) follows from $ h_y \sim \exp(1)$ while considering the activity factor $q$ of the interferers, (b) results from the probability generating functional (PGFL) \cite{HaeB2013} of the PPP $\Phi_B$ and expressing $\mathcal{C}_1$ as the union of $\mathcal{C}_3$ and $\mathcal{C}_1 \setminus \mathcal{C}_3$, where $\mathcal{C}_3 = \mathbf{b}(0,z_1)$ and $z_1 = {\tt max}(0,r_1-v)$ (See Fig. \ref{Fig:Circle Model}), and (c) follows by splitting the integral into the two regions followed by converting the integral from Cartesian to polar coordinates and using the law of cosines in which $r^2+v^2-2rv\cos\gamma = r_1^2$ (see Fig. \ref{Fig:Circle Model}). The lower limit of the integration region of $\mathcal{C}_1 \setminus \mathcal{C}_3$ (dark shaded region in Fig. \ref{Fig:Circle Model}) takes values of $r_1-v$ and $v-r_1$ for scenarios 1 and 2  respectively, therefore we use the $|v-r_1|$ to capture both scenarios.
\subsection{Proof of Lemma \ref{Lem: Lap_I3}}
Proceeding similar to (a) in Appendix A, the conditional Laplace transform of $I_3$, $\notag \Lagr_{I_{3}|R_1,R_2}(s|r_1,r_2)$
\begin{align}
&\notag = \E \bigg[\prod\limits_{z \in \Phi_A \backslash (\mathcal{C}_1 \cup \mathcal{C}_2)} 1-q + \frac{q}{1+s{\|z\|}^{-\alpha }}\bigg]\\\notag
&\stackrel{(a)}= e^{- p_Aq\lambda \int\limits_{r_2}^{\infty} \frac{2\pi r \ud r}{1+\frac{r^{\alpha }}{s}} }\exp\bigg( p_Aq\lambda  \underbrace{\int\limits_{\mathcal{C}_1 \setminus \mathcal{C}_2 } \frac{2f(r,r_1) \: r\ud r}{1+\frac{{r}^{\alpha }}{s}}}_{\mathcal{B}(r_1,r_2,v)} \bigg)\bigg)  \notag
\end{align}
where (a) results from the PGFL of the PPP $\Phi_A$ and splitting the integral into two regions, (b) follows by converting the integral from Cartesian to polar coordinates and using the law of cosines with $f(r,r_1)$ defined in Lemma \ref{Lem: Lap_I2}. The integral in the second term of (b) depends on the integration region $\mathcal{C}_1 \setminus \mathcal{C}_2$ with its lower limit (denoted by $a$ in Fig. \ref{Fig:Circle Model}) taking values \{$v-r_1, r_2$\} for cases 1 and 2 . The integration region is zero for case 3 as $\mathcal{C}_1 \setminus \mathcal{C}_2 = \emptyset$ ($\mathcal{C}_1$ is engulfed inside $\mathcal{C}_2$). The integral $\mathcal{B}(r_1,r_2,v)$ is summarized below
\begin{align} \label{Eqn:B_integral}
\mathcal{B}(r_1,r_2,v) &= \left\{
     \begin{array}{ll}
       \int_{v-r_1}^{v+r_1} \frac{2f(r,r_1) \: r\ud r}  {1+\frac{{r}^{\alpha }}{s}},  & \text{case 1} \\
       \int_{r_2}^{v+r_1} \frac{2f(r,r_1) \: r\ud r}{1+\frac{{r}^{\alpha }}{s}},   & \text{case 2}\\
       0, & \text{case 3}
      \end{array}
   \right..
\end{align}
\bibliographystyle{IEEEtran}
\bibliography{ref-HD-Updated2}
\end{document}

%% file: notation.tex
\def\nb0{{\mathbf{0}}}
\def\nb1{{\mathbf{1}}}

\newtheorem{lemma}{Lemma}
\newtheorem{thm}{Theorem}

\newtheorem{theorem}{Theorem}

\newtheorem{cor}{Corollary}

\def\E{\mathbb{E}}

\def\P{\mathbb{P}}

\def\R{\mathbb{R}}

\def\sir{\mathtt{SIR}}

